\documentclass[11pt,oneside]{article}
\usepackage{import}
\usepackage{etoolbox}

\usepackage[font=small,skip=0pt]{caption}

\providebool{DRAFT}
\booltrue{DRAFT}

\usepackage{etoolbox}
\usepackage{xparse}

\makeatletter%
\@ifclassloaded{book}
{%
  \newcommand{\whenbook}[1]{%
    #1%
  }%
}
\makeatother%

\providecommand{\whenbook}[1]{}%

\providebool{bookdraft}%

\NewDocumentEnvironment{When}{m +b}{%
  \providebool{#1}%
  \ifbool{#1}{#2}{}%
}{}%

\NewDocumentEnvironment{Unless}{m +b}{%
  \providebool{#1}%
  \ifbool{#1}{}{#2}%
}{}%

\providebool{DRAFT}
\ifbool{DRAFT}{
  \newcommand{\whendraft}[1]{#1}%
    \newcommand{\unlessdraft}[1]{}%
  }

\NewDocumentEnvironment{Draft}{+b}{%
  \whendraft{#1}%
}{%
}%

\NewDocumentEnvironment{DRAFT}{O{BlueViolet}}{%
  \begin{Draft}%
    \color{#1}%
  }{%
  \end{Draft}
}%

\usepackage[hyphens]{url}

\usepackage[style=alphabetic,natbib=true,maxalphanames=4,minalphanames=3,maxnames=10,doi=false,isbn=false,url=false]{biblatex} %

\AtEveryBibitem{%
\ifentrytype{proceedings}{
    \clearfield{editor}%
    \clearname{editor}%
  }{}
  \ifentrytype{inproceedings}{
      \clearfield{editor}%
      \clearname{editor}%
}{}
}

\usepackage{amsmath}                               %
\usepackage{amsthm}                                %
\usepackage[anchorcolor=blue,colorlinks]{hyperref}%
\usepackage{varioref}
\usepackage{cleveref}

\crefname{invariantsi}{invariant}{invariants}

\usepackage[margin=1in]{geometry}

\usepackage{fancyhdr}

\usepackage{microtype}          %
\usepackage[T1]{fontenc}

\usepackage{article} %
\usepackage{math} %
\usepackage{algo} %
\usepackage{wrapfig} %
\usepackage{placeins} %
\usepackage{advdate}%

\usepackage{appendix}

\usepackage{graphicx} %
\usepackage{layout} %
\usepackage{setspace} %
\usepackage{mathtools} %
\usepackage{grffile} %
\usepackage{pdfpages} %
\usepackage{mdframed} %
\usepackage{bm} %
\usepackage{needspace} %
\usepackage{pbox} %
\usepackage{cancel}
\usepackage[normalem]{ulem}

\usepackage[full]{textcomp}     %
\usepackage{lmodern}            %
\usepackage[T1]{fontenc}        %
\usepackage{inconsolata}        %

\usepackage{titlesec}

\usepackage{truncate}

\NewDocumentCommand{\undernote}{s O{blue} m m}{
  \IfBooleanT{#1}{\smash}%
  {\color{#2} %
    \underbrace{\normalcolor%
      #4}_{\mathclap{\text{#3}}} %
  }%
  \IfBooleanT{#1}{\vphantom{#4}}
}

\newcommand{\defterm}[1]{{\boldmath\normalfont \bfseries #1}}%
\renewcommand{\defterm}{\emph}%

\makeatletter
\g@addto@macro\bfseries{\boldmath}
\makeatother

\usepackage{figs}

\titlespacing*{\paragraph}{%
  0pt}{
  {\medskipamount}}{
  1em}

\titleformat{\subparagraph}[runin]{\itshape}{0pt}{}{}%

\titlespacing*{\subparagraph}{%
  0pt}{
  {\medskipamount}}{
  1em}

\newcommand{\NameColorComment}[3]{%
  \whendraft%
  {\renewcommand\thefootnote{\textcolor{#2}{\arabic{footnote}}}%
    \footnote{\color{#2}#1: #3}%
  }%
}%

\newlength{\myalphabet}                %
\newlength{\mywidth}                   %
\newlength{\mymargin}                  %

\providecommand{\wrt}{with respect to\xspace}%
\usepackage{skak}%

\bibliography{references}

\providecommand{\poly}{\fparnew{\operatorname{poly}}}

\definecolor{calypso}{RGB}{50, 104, 145} %

\hypersetup{allcolors=calypso}  %

\definecolor{almostblack}{RGB}{18, 18, 18} %

\color{almostblack}%
\hypersetup{allcolors=almostblack}  %

\providebool{article}
\booltrue{article}

\begin{document}

\newcommand{\W}{\mathcal{W}}
\newcommand{\potential}{\fparnew{\phi}} %
\providecommand{\pot}{\potential}       %
\providecommand{\pote}{\potential}
\providecommand{\len}{\fparnew{\ell}}       %
\newcommand{\potlen}{\fparnew{\ell_{\potential}}} %
\newcommand{\plen}{\fparnew{\ell_{\potential}}}   %
\newcommand{\p}{\fparnew{\varphi}}                %
\newcommand{\lenp}{\fparnew{\ell_{\p}}}

\newcommand{\through}{\fparnew{\operatorname{T}}} %
\newcommand{\sandwich}{\fparnew{\operatorname{S}}} %
\newcommand{\disp}{\fparnew{d_{\p}}}                 %
\newcommand{\dis}{\fparnew{d}}                     %
\newcommand{\hopd}[1]{\fparnew{d^{#1}}}            %
\newcommand{\hopdp}[2][\p]{\fparnew{d^{#2}_{#1}}}            %
\newcommand{\shopd}[1]{\fparnew{\hat{d}^{#1}}}               %
\newcommand{\shopdp}[2][\p]{\fparnew{\hat{d}_{#1}^{#2}}}
\providecommand{\deg}{\fparnew{\operatorname{deg}}} %
\providecommand{\indeg}{\fparnew{\operatorname{deg}^-}} %
\providecommand{\outdeg}{\fparnew{\operatorname{deg}^+}} %
\providecommand{\hopindeg}[1]{\fparnew{\operatorname{deg}_{-}^{#1}}} %
\providecommand{\hopoutdeg}[1]{\fparnew{\operatorname{deg}_{-}^{#1}}} %
\providecommand{\heavyin}[1]{H^{-}_{#1}} %
\providecommand{\hin}{\heavyin} %
\providecommand{\heavyout}[1]{H^{+}_{#1}} %
\providecommand{\hout}{\heavyout}
\providecommand{\reachin}[1]{R^{-}_{#1}} %
\providecommand{\reachout}[1]{R^+_{#1}}  %
\providecommand{\apxheavy}{\smash{\tilde{H}}} %
\providecommand{\negV}{N}                     %
\providecommand{\negE}{E^-}

\renewcommand{\subsection}[1]{\paragraph{#1.}}

\newcommand{\yufan}{\NameColorComment{Yufan}{Maroon}}
\newcommand{\kent}{\NameColorComment{Kent}{Cerulean}}
\newcommand{\peter}{\NameColorComment{Peter}{PineGreen}}

\newcommand{\authornote}{\texttt{$\setof{\texttt{huan1754},\texttt{jin453},\texttt{krq}}$@purdue.edu}. Purdue
  University, West Lafayette, Indiana. YH was supported in part by NSF
  grant IIS-2007481. PJ and KQ were supported in part by NSF grant
  CCF-2129816.}

\author{Yufan Huang \and Peter Jin \and Kent Quanrud}


\title{Faster single-source shortest paths with
  negative~real~weights~via~proper~hop~distance\footnote{\authornote}}

\maketitle

\begin{abstract}
  The textbook algorithm for single-source shortest paths with
  real-valued edge weights runs in $\bigO{m n}$ time on a graph with
  $m$ edges and $n$ vertices. A recent breakthrough algorithm by
  \citet{Fineman24} takes $\apxO{m n^{8/9}}$ randomized time. We
  present an $\apxO{m n^{4/5}}$ randomized time algorithm building on
  ideas from \cite{Fineman24}.
\end{abstract}

\section{Introduction}

Let $G = (V,E)$ be a directed graph with $m$ edges, $n$ vertices, and
real-valued edge lengths $\len: E \to \reals$. For a fixed source
vertex $s$, the single-source shortest path (SSSP) problem is to
compute either a negative length cycle in $G$, or the shortest-path
distance from $s$ to every other vertex in $G$. The textbook
$\bigO{m n}$ time dynamic programming algorithm; discovered in the
1950's by \citet{Shimbel55}, \citet{Ford56}, \citet{Bellman58}, and
\citet{Moore59}; is a staple of introductory algorithms courses. Until
very recently, it was also the fastest algorithm for this problem.

There are much faster algorithms for important special cases. When the
edge weights are nonnegative, Dijkstra's algorithm takes
$\bigO{m + n \log n}$ time \cite{Dijkstra59,FT87}.  There is also a
long line of research when the edge weights are integer
\cite{GT89,Goldberg95,Sankowski05,YZ05,CMSV17,BLN+20,AMV20,CKLPGS22,BNW22,BCF23},
leading to a remarkable nearly linear running time in
\cite{BNW22}. These algorithms all have at least a logarithmic
dependence on the magnitude of the most negative edge weight.

The textbook bound of $\bigO{mn}$ held its ground until a recent and
exciting breakthrough by \citet{Fineman24}, who gave an
$\apxO{m n^{8/9}}$ randomized time algorithm in the Real RAM
model. ($\apxO{\cdots}$ hides $\log{n}$ factors.) Many of the
interesting ideas in this work are detailed below in
\Refsection{preliminaries}.

This work was inspired by \cite{Fineman24}. Building on Fineman's
ideas, we obtain:
\begin{restatable}{theorem}{TheTheorem}
  There is a (Las Vegas) randomized algorithm solving the SSSP problem
  for real-weighted graphs that runs in
  \begin{math}
    \bigO{m n^{4/5} \log^{2/5} n + n^{9/5} \log^{7/5} n}
  \end{math}
  randomized time with high probability.
\end{restatable}
The algorithm follows the same general framework laid out by
\cite{Fineman24}, which eliminates negative edges incrementally along
randomly sampled ``independent sets'' or ``sandwiches'' of negative
edges. An interesting idea that we found useful was to focus on
``proper'' walks, where all the negative edges are required to be
distinct, and on ``proper hop distances'', which consider proper walks
with an exact number of negative edges (or ``hops'').  Morally, it is
much harder for a random sample of negative edges to preserve a
negative proper hop distance between two vertices, especially when the
number of hops is large. This fact allows us to aggressively sample
larger ``sandwiches'' for the same amount of effort. We explain in
greater detail below.

\section{Preliminaries}

\labelsection{preliminaries}

Let $G = (V,E)$ be a directed graph with $m$ edges, $n$ vertices, and
real-valued edge lengths $\len: E \to \reals$. We let $k$ denote the
number of negative edges in $G$. The distance from $u$ to $v$ is
defined as the infimum length over all walks from $u$ to $v$.  We let
$d(u,v)$ denote the distance from $u$ to $v$ \wrt $\len$. More
generally, for vertex sets $S$ and $T$, we let $d(S,T)$ denote the
minimum distance $d(s,t)$ over all $s \in S$ and $t \in T$.

\subsection{Vertex potentials}
Given vertex potentials $\p: V \to \reals$, let
$\lenp{e} \defeq \len{e} + \p{u} - \p{v}$ for an edge $e =
(u,v)$. \citet{Johnson77} observed that the potentials
$\p{v} = d(V,v)$ give $\lenp{e} \geq 0$ for all $e$ (assuming there
are no negative-length cycles). We let $G_{\varphi}$ denote the graph
with edges reweighted by $\p$ and $\disp$ denote distances \wrt
$\lenp$.

We say that potentials $\p$ are \emph{valid} if $\lenp{e} \geq 0$ for all
$e$ with $\len{e} \geq 0$; i.e., $\p$ does not introduce any new negative
edges. Given two valid potentials $\p_1$ and $\p_2$, $\max{\p_1,\p_2}$
and $\min{\p_1,\p_2}$ are also valid potentials. If $\p_1$ is valid
for $G$, and $\p_2$ is valid for $G_{\p_1}$, then $\p_1 + \p_2$ is
valid for $G$.

We say that $\p$ \emph{neutralizes} a negative edge $e$ if
$\lenp{e} \geq 0$. This work, like other previous work since
\cite{Goldberg95}, iteratively computes and integrates valid
potentials $\p$ that neutralize more and more negative edges, until
all edges are nonnegative.

\subsection{Preprocessing and negative vertices}

By standard techniques, we may assume the maximum in-degree and
out-degree are both $\bigO{m / n}$.  We also preprocess the input
graph $G$ as follows \cite{Fineman24}.  For each vertex $v$, let
$\mu_v = \min{\len{v,x} \where (v,x) \in E}$ be the minimum length of
the edges leaving $v$. Consider the graph $G'$ where we:
\begin{algorithm}
\item Replace each vertex $v$ with two vertices $v^-,v^+$, and an arc
  $(v^-,v^+)$ with length $\mu_v$.
\item Replace each arc $(u,v)$ with the arc $(u^+,v^-)$ with length
  $\len{u,v} - \mu_u$.
\end{algorithm}
$G'$ has similar size to $G$, and distances in $G'$ give distances in
$G$. Additionally, all the negative edges in $G'$ have the form
$(v^-,v^+)$ for some vertex $v$. So there are at most $n/2$ negative
edges, and each negative edge $(u,v)$ is the unique outgoing edge of
its tail $u$.

Henceforth we assume that $G$ has at most $n/2$ negative edges, and
that each negative edge $(u,v)$ is the unique incoming edge of its
head $v$ and the unique outgoing edge of its tail $u$. We identify
each negative edge $(u,v)$ with its tail $u$, and call $u$ a
\defterm{negative vertex}. When there is no risk of confusion, we may
reference a negative edge by its negative vertex and vice-versa. For
example, we say we \emph{neutralize a negative vertex} when we mean we
neutralize its associated negative edge.

We let $N$ denote the set of negative vertices in $G$.  For a set of
negative vertices $S$, we let $G_S$ denote the subgraph obtained by
restricting the set of negative edges to those corresponding to $S$.

When computing Johnson's potentials $\p{v} = \dis{V,v}$, any
nontrivial shortest walk to $v$ might as well start with a negative
edge. Therefore $\dis{V,v} = \min{0, \dis{\negV,v}}$ for all $v$.

\subsection{Hop distance}

The number of \emph{hops} in a walk is the number of negative edges
along the walk, counted with repetition. An \defterm{$r$-hop walk} is
a walk with at most $r$ negative edges. For an integer $r$, and
vertices $s,t \in V$, we let
\begin{math}
  \hopd{r}{s,t}
\end{math}
denote the infimum length over all $r$-hop walks from $s$ to $t$.
Hop distances obey the recurrence
\begin{align*}
  \hopd{r+1}{s,t} = \min{\hopd{r}{s,t}, \min_{u \in N} \hopd{r}{s,u} +
  \hopd{1}{u,t}}. \labelthisequation{hop-distance}
\end{align*}
For fixed $s$, if $\hopd{r+1}{s,v} = \hopd{r}{s,v}$ for all
$v \in \negV$, then $\dis{s,v} = \hopd{r+1}{s,v}$ for all $v$. Also,
for $r \in \naturalnumbers$ and any set of vertices $S$, the
potentials $\p{v}_1 = \hopd{r}{S,v}$ and $\p{v}_2 = -\hopd{r}{v, S}$
are valid.

For fixed $s$, and $r \in \naturalnumbers$, one can compute the
$r$-hop distance $\hopd{r}{s,v}$ for all $v \in V$ by a hybrid of
Dijkstra's algorithm and Bellman-Ford-Moore-Shimbel in
$\bigO{r m + r n \log n}$ time \cite{DI17,BNW22}. The algorithm can be
inferred from \refequation{hop-distance}: given $\hopd{r}{s,u}$ for
all $u$, we can compute $\min_{u \in N} \hopd{r}{s,u} + \hopd{1}{u,t}$
for all $t$ with a single call to Dijkstra's algorithm over an
appropriate auxiliary graph. By maintaining parent pointers, we can
also obtain the walks supporting the $r$-hop distances.

More generally, hop distance can be defined \wrt any set of vertices
$S$ that contains $\negV$, and the observations above would still
hold. We need this generalization for the following reason. Our
algorithm, following \cite{Fineman24}, reweights the graph along
multiple valid potential functions in a single iteration. The initial
potential functions in an iteration are not intended to neutralize any
edges but rather massage the graph for subsequent steps. However, they
may incidentally neutralize some negative edges, which can decrease
the hop counts of walks and disrupt various assumptions about hop
distances. To avoid this, like \cite{Fineman24}, we freeze the set of
negative vertices at the beginning of each iteration. During an
iteration, we continue to refer to a negative vertex as negative even
if its associated edge becomes nonnegative. This ensures the identity
$\hopdp{h}{u,v} = \hopd{h}{u,v} + \p{u} - \p{v}$ for any potential
$\p$ within an iteration.

\subsection{Remote edges, betweenness, and sandwiches}

We say that one vertex $u$ can \emph{negatively reach} another vertex
$v$ if there is a negative length walk from $u$ to $v$. More
specifically, we say that $u$ can negatively reach $v$ in $h$ hops if
there is a negative walk from $u$ to $v$ with at most $h$ negative
edges.

For a parameter $r \in \naturalnumbers$, and a set of negative
vertices $U \subseteq V$, $U$ is \defterm{$r$-remote} if $U$ can reach
at most $n /r$ vertices via negative $r$-hop paths. They are
important because they can be neutralized efficiently.

\begin{lemma}[{\cite[Lemma 3.3]{Fineman24}}]
  \labellemma{neutralize-remote} Let $G$ have maximum in-degree and
  out-degree $\bigO{m / n}$.  There is an algorithm that, given a set
  of $r$-remote vertices $U$, returns either a set of potentials
  neutralizing $U$ or a negative cycle in
  $\bigO{(r + \sizeof{U} / r)(m + n \log n)}$
  time.
\end{lemma}
(As stated, \cite[Lemma 3.3]{Fineman24} reports the existence of a
negative cycle, but the algorithm can easily be amended to return the
negative cycle. The same comment applies to
\Reflemma{betweenness-reduction} below.)

To generate $r$-remote sets, \cite{Fineman24} introduced the notions
of betweenness and negative sandwiches.  For a parameter
$r \in \naturalnumbers$, and $s,t \in V$, the
\defterm{$r$-betweenness} is defined as the number of vertices $v$
such that $d^r(s,v) + d^r(v,t) < 0$. \cite{Fineman24} gave the following
randomized procedure to reduce the betweenness of all pairs of
vertices.

\begin{lemma}[{\cite[Lemma 3.5]{Fineman24}}]
  \labellemma{betweenness-reduction} There is an
  $\bigO{r b m \log n + r b n \log^2 n + b^2 n \log^2 n }$
  time randomized algorithm that returns either a set of valid
  potentials $\p$ or a negative cycle.  With high probability, all
  pairs $(s,t) \in V \times V$ have $r$-betweenness (at most) $n/b$
  \wrt $\ell_{\varphi}$.
\end{lemma}

\cite{Fineman24} defines a \defterm{negative sandwich} as a triple
$(s,U,t)$, where $s,t \in V$, and $U \subseteq N$, such that
$\hopd{1}{s,u} < 0$ and $\hopd{1}{u,t} < 0$ for all $u \in U$.  We
generalize this as follows.  For a parameter $h$, a \defterm{weak
  $h$-hop negative sandwich} is a triple $(s,U,t)$ such that
\begin{math}
  \hopd{h}{s,u} + \hopd{h}{u,t} \leq 0
\end{math}
for all $u \in U$.  The following lemma, converting sandwiches into
remote edges when the betweenness is low, generalizes \cite[Corollary
3.9]{Fineman24} from negative sandwiches to weak multi-hop negative
sandwiches.
\begin{lemma}
  \labellemma{sandwich->remote} Let $(s,U,t)$ be a weak $h$-hop
  negative sandwich and let $(s,t)$ have $(r+h)$-betweenness
  $n/r$. Then one can compute valid potentials $\p$ making $U$
  $r$-remote in $\bigO{(r+h)(m + n \log n)}$ time.
\end{lemma}
\begin{proof}
  Without loss of generality, we assume all the vertices are reachable
  from $s$. (Given potentials making $U$ $r$-remote in the subgraph
  reachable from $s$, assign the remaining vertices the maximum of
  these potentials.)  Define
  \begin{math}
    \p{v} \defeq \max{d^{r+h}(s,v), - d^{r+h}(v,t)}
  \end{math}
  for $v \in V$. For all $u \in U$,
  \begin{align*}
    -d^{r+h}(u,t) \geq - d^h(u,t) \tago{\geq} d^h(s,u) \geq d^{r+h}(s,u),
  \end{align*}
  where \tagr is because $d^h(s,u) + d^h(u,t) \leq 0$.  Thus
  \begin{math}
    \p{u} = - d^{r+h}(u,t) \geq d^h(s,u)
  \end{math}
  for all $u \in U$.

  For $v \in V$ with $d^{r+h}(s,v) + d^{r+h}(v,t) \geq 0$, since
  $d^{r+h}(s,v) \geq - d^{r+h}(v,t)$, we have
  \begin{align*}
    \p{v} = d^{r+h}(s,v).
  \end{align*}

  Now, let $u \in U$ and $v \in V$ with $d^{r+h}(s,v) + d^{r+h}(v,t)
  \geq 0$.
  We have
  \begin{align*}
    \disp{u,v}^r                  %
    &=               %
      \dis{u,v}^r + \p{u} - \p{v}
      \geq
      \dis{u,v}^r + \dis{s,u}^h - \dis{s,v}^{r+h}
      \geq 0
  \end{align*}
  by the triangle inequality.

  Thus a vertex $v$ is negatively reachable in $r$ steps from $U$ only
  if $d^{r+h}(s,v) + d^{r+h}(v,t) < 0$. There are only $n/r$ such
  vertices, making $U$ $r$-remote.  As for the running time, it takes
  $\bigO{(r+h)(m + n \log n)}$ time to compute $\p{v}$ for all $v$.
\end{proof}

\subsection{Sketch of \cite{Fineman24}} To help motivate the
preliminaries above, and to place ensuing ideas in their proper
context, we briefly sketch the existing $\apxO{m n^{8/9}}$ randomized
time algorithm of \cite{Fineman24}.

Call a set of negative vertices $U$ \emph{independent} if
$\hopd{1}{U,x} \geq 0$ for all $x \in U$. If $U$ is
independent, then in the graph restricting the negative vertices to
$U$, we have $d(V,x) = \hopd{1}{V,x}$ for all $x$. Thus an independent
set can be neutralized by Johnson's technique in nearly linear time.

Suppose there are currently $k$ negative vertices and edges in $G$.
\cite{Fineman24} gives a randomized algorithm that produces either an
independent set or a negative sandwich of size
$\bigOmega{k^{1/3}}$. The independent set can be neutralized directly
in nearly linear time. To neutralize the negative sandwich with
\Reflemma{neutralize-remote}, we need the endpoints of the sandwich to
have low betweenness. Thus \cite{Fineman24} first spends
$\apxO{m k^{2/9}}$ time to decrease the $k^{1/9}$-betweenness to
$n/k^{1/9}$ for all pairs of vertices, before sampling the independent
set or negative sandwich.  Then, if the sample returns a negative
sandwich, the negative sandwich is neutralized in
$\apxO{m k^{2/9}}$ time. One way or another, it takes
$\apxO{m k^{2/9}}$ time to neutralize $k^{1/3}$ negative edges.  It
takes $\apxO{m n^{8/9}}$ randomized time overall to repeat these steps
until there are no negative edges remaining.

\section{Proper multi-hop walks}

Let $h \in \naturalnumbers$ be a parameter to be determined.  Consider
the following stricter notion of hops and hop distance.  We define a
\emph{proper $h$-hop walk} as a walk with exactly $h$ negative edges
and where all negative vertices are distinct.  For vertices $s,t$, let
\begin{math}
  \shopd{h}{s,t}
\end{math}
denote the infimum length over all proper $h$-hop walks from $s$ to
$t$.  We call
\begin{math}
  \shopd{h}{s,t}
\end{math}
the \defterm{proper $h$-hop distance} from $s$ to $t$.

Computing proper hop distances is tricky because negative vertices
cannot repeat. Even if there are no negative cycles, and all distances
are attained by paths, standard algorithms will still want to reuse
negative edges in order to meet the hop requirement. For small values
of $h$, single-source proper $h$-hop distances can be computed with
high probability in $\apxO{2^{\bigO{h}} m}$ randomized time using
color-coding techniques \cite{AYZ95}. Fortunately, we can sidestep
computing proper $h$-hop distances altogether, and take $h$ to be much
larger, with the following lemma.

\begin{lemma}
  \labellemma{implicit-proper-distance} There is an
  $\bigO{h \parof{m + n \log n}}$-time algorithm that, given a set of
  negative vertices $S$, returns either a negative cycle, a pair
  $s,t \in S$ with $\shopd{h}{s,t} < 0$, or the distances $d(V,t)$ in
  $G_S$ for all $t \in V$.
\end{lemma}

\begin{proof}
  Consider the subgraph $G_S$. We first compute $\hopd{h-1}{S,t}$ and
  $\hopd{h}{S,t}$ (and the supporting walks) for all $t\in V$. If
  $\hopd{h}{S,t}=\hopd{h-1}{S,t}$ for all $t\in S$, then
  $\dis{V,t}=\min{0,\dis{S,t}}=\min{0,\hopd{h}{S,t}}$ for all $t\in V$.

  Otherwise, $\hopd{h}{S,t} < \hopd{h-1}{S,t} \leq 0$ for some
  $t\in S$. Consider the $h$-hop walk attaining $\hopd{h}{S,t}$. We
  know this walk has exactly $h$ negative edges because it has length
  less than $\hopd{h-1}{S,t}$.  If some negative vertex repeats, then
  the cycle along the walk between occurrences must be a negative
  cycle, since otherwise removing the cycle would give an $(h-1)$-hop
  walk to $t$ with length less than $\hopd{h-1}{S,t}$.  If the
  negative vertices are distinct, then we have a proper $h$-hop walk
  to $t$ with negative length, as desired.
\end{proof}

\section{Sampling weak negative sandwiches}
\labelsection{sampling}

Recall that weak negative sandwiches can be efficiently neutralized by
the techniques in \Refsection{preliminaries}. Similar to
\cite{Fineman24}, we try to find large sandwiches quickly by random
sampling. We leverage proper $h$-hop distances to increase the size of
the sandwich by a factor of $\Omega(\sqrt{h})$.

\begin{lemma}
  \labellemma{sampling} For $h = \bigOmega{\log n}$, there is a
  randomized algorithm that, in $\bigO{h \parof{m + n \log n}}$
  randomized time, returns either a negative cycle, a weak $h$-hop
  sandwich $(s,U,t)$, or a set of negative vertices $S$ and the
  distances $d(V,v)$ for all vertices $v$ in the subgraph $G_S$. With
  high probability, we have
  $\sizeof{U}, \sizeof{S} \geq \bigOmega{\sqrt{h k}}$ (when they are
  returned by the algorithm).
\end{lemma}
\begin{proof}
  Let $q = 2\sqrt{k/h}$.  Let $S \subseteq N$ randomly sample each
  negative vertex independently with probability $1 / q$.  $S$ has at
  least $k/2q = \sqrt{k h} / 4$ negative vertices with high
  probability by the multiplicative Chernoff bound.  Applying
  \Reflemma{implicit-proper-distance} to $S$, we obtain either (a) a
  negative cycle, (b) the distances $d(V,t)$ for all vertices $t$ in
  $G_S$, or (c) a pair of $s,t \in S$ with $\shopd{h}{s,t} < 0$. The
  negative cycle and the distances are returned immediately.  In the
  third case, we return the sandwich $(s,B_{s,t},t)$, where
  \begin{math}
    B_{s,t} \defeq \setof{u \in \negV \where \hopd{h}{s,u} +
      \hopd{h}{u,t} < 0},
  \end{math}
  in $\bigO{h \parof{m + n \log n}}$ time.

  We say that an ordered pair $(s,t)$ is \emph{sampled by $S$} if
  $s \in S$, $t \in S$, and $\shopd{h}{s,t} < 0$ in $G_S$.  Call
  $(s,t)$ \emph{weak} if $\sizeof{B_{s,t}} \leq \sqrt{kh}$, and
  \emph{strong} otherwise. We claim that with high probability no weak
  pairs are sampled by $S$.

  This is where proper hop distance comes to the fore: any proper
  $h$-hop walk from $s$ to $t$ with negative length must include at
  least $h-1$ vertices from $B_{s,t} \setminus \setof{s,t}$. For fixed
  $s,t$, $\sizeof{B_{s,t} \cap S}$ is a sum of independent
  $\setof{0,1}$-random variables.
  If $(s,t)$ is weak, then this sum has mean at most $h/2$. Since $h =
  \bigOmega{\log n}$, by the multiplicative Chernoff bound,
  $\sizeof{B_{s,t} \cap S} \leq h$, hence $(s,t)$ is not sampled, with
  high probability.

  By the union bound, $S$ samples at least $\sqrt{kh}/4$ negative
  vertices and no weak pairs with high probability. Then $\sizeof{S}$
  and $\sizeof{B_{s,t}}$ (when available) both have the desired size.
\end{proof}

\section{Putting it all together}

We now present the overall algorithm and analysis, combining
ingredients from \Refsection{preliminaries} with the new sampling
techniques from \Refsection{sampling}.

\TheTheorem*
\begin{proof}
  The algorithm iteratively tries to compute valid potentials $\p$
  that neutralizes some of the remaining negative edges, until no
  negative edges remain.  We describe a single iteration.

  Suppose there are $k$ negative edges. We assume that
  $k = \bigOmega{\log^7{n}}$, since otherwise we can use Johnson's
  technique with $k$-hop distances in $\bigO{k \parof{m + n \log n}}$
  time.  Let $h = k^{1/5} / \log^{2/5} n = \bigOmega{\log n}$.  By
  \Reflemma{betweenness-reduction}, in
  $\bigO{h^2 (m \log n + n \log^2 n )}$ time, we compute either a
  negative cycle (and halt) or compute valid potentials $\p_1$ that,
  with high probability, reduce the $\bigO{h}$-betweenness of all
  pairs of vertices to $n / h$. Henceforth we work in $G_{\p_1}$.

  Next we invoke \Reflemma{sampling}. If it returns a negative cycle
  then we halt.  If it returns a set of negative vertices $S$, and all
  the Johnson potentials $\p_2(v) = d(V,v)$ in $G_S$, then we re-weight
  the graph with $\p_2$ and neutralize $S$.

  Otherwise we obtain a weak $h$-hop negative sandwich $(s,U,t)$.  If
  the betweenness-reduction succeeded, then by
  \Reflemma{sandwich->remote}, we either return a negative cycle or
  compute potentials $\p_2$ that make $U$ $h$-remote in
  $\bigO{h (m + n \log n)}$ time. Note that we can test if $U$ is
  indeed $h$-remote in the same amount of time.  If not, then the
  betweenness-reduction failed, and we restart the
  iteration. Otherwise, by \Reflemma{neutralize-remote}, we either
  find a negative cycle (and halt) or compute potentials $\p_3$
  neutralizing $U$ (in $G_{\p_1 + \p_2}$) in
  $\bigO{\parof{h + \sizeof{U} / h} \parof{m + n \log n}}$ time.

  Altogether, with high probability of success, and in
  $\bigO{\parof{h^2 \log n + \ell / h} (m + n \log n)}$ randomized
  time, we either neutralize $\ell = \bigOmega{\sqrt{kh}}$ negative
  edges, or return a negative cycle.

  Starting from at most $n$ negative vertices, it takes
  $\bigO{\sqrt{n/h} \log n}$ successful iterations to reduce the
  number of negative vertices to $\bigO{\log^7 n}$.  By the union
  bound, all the iterations succeed with high probability.

  The running time is dominated by the time spent to reduce the number
  of negative vertices by half. Thus the overall running time is
  $\bigO{\parof{n^{1/2} h^{3/2} \log n + n / h} \parof{m + n \log n}}
  = \bigO{m n^{4/5} \log^{2/5} n + n^{9/5} \log^{7/5} n}$ with high
  probability.
\end{proof}

\printbibliography[nottype=proceedings]


\end{document}
